\documentclass[referee,usenatbib,useAMS]{biom}
\usepackage{amsmath}
\usepackage{hyperref}
\usepackage{amssymb}

\renewcommand{\P}{{\mathbb P}}

\newtheorem{proposition}{Proposition}
\begin{document}

\title[Discussion on \textquotedblleft Model Confidence Bounds for Variable
Selection\textquotedblright\ by Li et al.]{Discussion on \textquotedblleft Model Confidence Bounds for Variable
Selection\textquotedblright\ by Yang Li, Yuetian Luo, Davide Ferrari,
Xiaonan Hu, and Yichen Qin}
\author{Hannes Leeb$^{\text{*,+}}$, Benedikt M. P\"{o}tscher$^{\text{*}}$,
and Danijel Kivaranovic$^{\text{*}}$ \\
$^{\text{*}}$Department of Statistics, University of Vienna\\
$^{\text{+}}$DataScience@UniVienna\\
e-mail: benedikt.poetscher@univie.ac.at}
\date{July 2018}
\maketitle

We congratulate the authors of \citet{Li18a} for their interesting paper,
and we thank the Editor for the opportunity to comment on it. We shall offer
some criticism of the proposed model confidence set in the following. Unless
noted otherwise, we use the same notation as in \citet{Li18a}. [\citet{Li18a}
are not explicit about the nature of the regressor variables. We read their
paper as considering nonstochastic regressors, and we note that similar
results are possible for models with stochastic regressors. Also, a full
column-rank assumption on the regressor matrix $X$ seems to be missing,
without which the symbol $m^{\ast }$ is not well-defined. We hence add this
assumption.]

Let us first strengthen Theorem 1, the main result of \citet{Li18a}, a
little bit. Recall that $\hat{r}(m_{1},m_{2})=B^{-1}\sum_{b=1}^{B}I(m_{1}%
\subseteq \hat{m}^{(b)}\subseteq m_{2})$ and define $\tilde{r}%
(m_{1},m_{2})=P\left( m_{1}\subseteq \hat{m}^{(b)}\subseteq m_{2}\Vert
Y\right) $, where $Y$ are the data. That is, $\tilde{r}(m_{1},m_{2})$ is the
analogue of $\hat{r}(m_{1},m_{2})$, computed from the exact bootstrap
distribution rather than from an approximation to it based on an i.i.d.
bootstrap sample of size $B$ (note that $\tilde{r}(m_{1},m_{2})=E(\hat{r}%
(m_{1},m_{2})\Vert Y)$, and thus is the quantity that $\hat{r}(m_{1},m_{2})$
is trying to approximate with the help of the bootstrap sample). Let $\tilde{%
m}_{L}$ and $\tilde{m}_{U}$ now be obtained from the exact bootstrap, i.e., $%
\tilde{m}_{L}$ and $\tilde{m}_{U}$ are defined by program (2) in %
\citet{Li18a}, but with $\tilde{r}(m_{1},m_{2})$ replacing $\hat{r}%
(m_{1},m_{2})$ everywhere. To exclude trivial cases, we assume that the
nominal confidence level satisfies $0<1-\alpha <1$. [Recall that $\hat{m}%
_{L} $ and $\hat{m}_{U}$ are defined by program (2) using $\hat{r}%
(m_{1},m_{2})$ as given in the paper.] We note that inspection of the proof
of Theorem 1 in \citet{Li18a} reveals that this proof actually seems to be
given for $\tilde{m}_{L}$ and $\tilde{m}_{U}$ rather than for $\hat{m}_{L}$
and $\hat{m}_{U}$ (although the theorem also holds for the latter as will
transpire from the subsequent result). We are now ready for the improved
version of Theorem 1 in \citet{Li18a}.

\begin{proposition}
\label{p1} Assume: (A.1) (Model selection consistency) $P(\hat{m}\neq
m^{\ast })=o(1)$; (A.2) (Bootstrap validity) For the re-sampled model $\hat{m%
}^{(b)}$, assume $P(\hat{m}^{(b)}\neq \hat{m})=o(1)$.

(a) Then $P(\tilde{m}_{L}\subseteq m^{\ast }\subseteq \tilde{m}_{U})=1+o(1)$
and $P(\left\vert \tilde{m}_{U}\right\vert -\left\vert \tilde{m}%
_{L}\right\vert =0)=1+o(1)$; in fact, $P(\tilde{m}_{U}=\tilde{m}_{L}=\hat{m}%
=m^{\ast })=1+o(1)$ holds. Moreover, assuming (A.1) only, this statement
trivially continues to hold if $\tilde{m}_{L}$ and $\tilde{m}_{U}$ are both
replaced by $\hat{m}$.

(b) Part (a) also holds if we replace $\tilde{m}_{L}$ by $\hat{m}_{L}$ and $%
\tilde{m}_{U}$ by $\hat{m}_{U}$ (where the number of bootstrap replications $%
B$ may depend on $n$).
\end{proposition}

\begin{proof}
(a) (A.2) implies that $\tilde{r}(\hat{m},\hat{m})=P\left( \hat{m}^{(b)}=%
\hat{m}\Vert Y\right) $ converges to $1$ in probability as $n\rightarrow
\infty $. [To see this note that (A.2) implies $1-P(\hat{m}^{(b)}=\hat{m}%
)=E(1-P(\hat{m}^{(b)}=\hat{m}\Vert Y))\rightarrow 0$ as $n\rightarrow \infty 
$. But this shows that $E(\left\vert 1-P(\hat{m}^{(b)}=\hat{m}\Vert
Y)\right\vert )$ converges to zero, which even implies $L_{1}$-norm
convergence of $P(\hat{m}^{(b)}=\hat{m}\Vert Y)$ to $1$.] Obviously, this in
turn also implies $\sum\nolimits_{m\neq \hat{m}}\tilde{r}(m,m)\rightarrow 0$
in probability. Inspection of (2) in \citet{Li18a} (with $\tilde{r}%
(m_{1},m_{2})$ replacing $\hat{r}(m_{1},m_{2})$ everywhere) then shows that $%
\tilde{m}_{L}=\tilde{m}_{U}=\hat{m}$ holds on an event that has probability
converging to $1$. In view of (A.1), the event where $\tilde{m}_{U}=\tilde{m}%
_{L}=\hat{m}=m^{\ast }$ holds then also has probability converging to $1$.
This proves everything except for the last claim. For the last statement
just note that the event $\{\hat{m}=m^{\ast }\}$ trivially has probability
converging to one under (A.1).

(b) Consider first the case where $B$ does not depend on $n$: Then, under
(A.2), the event $A_{n}=\{\hat{m}=\hat{m}^{(1)}=\ldots =\hat{m}^{(B)}\}$,
which is the finite intersection $\bigcap_{b=1}^{B}\left\{ \hat{m}^{(b)}=%
\hat{m}\right\} $, has probability converging to one. Note that $\hat{r}(%
\hat{m},\hat{m})=1$ and $\sum\nolimits_{m\neq \hat{m}}\hat{r}(m,m)=0$ hold
on $A_{n}$. Inspection of (2) in \citet{Li18a} hence shows that $\hat{m}_{L}=%
\hat{m}_{U}=\hat{m}$ holds on $A_{n}$, which, together with (A.1), proves
the first claim in Part (b) for fixed $B$. Consider next the case where $%
B=B_{n}$ diverges to infinity as $n\rightarrow \infty $: Similarly as in
Part (a), it suffices to show that $\hat{r}(\hat{m},\hat{m})$ converges to $%
1 $ in probability (noting that $\sum\nolimits_{m\neq \hat{m}}\hat{r}%
(m,m)\rightarrow 0$ in probability then follows). Note that 
\begin{eqnarray*}
\hat{r}(\hat{m},\hat{m}) &=&B_{n}^{-1}\sum_{b=1}^{B_{n}}\left\{ I(\hat{m}%
^{(b)}=\hat{m})-E(I(\hat{m}^{(b)}=\hat{m})\Vert Y)\right\} +E(I(\hat{m}%
^{(b)}=\hat{m})\Vert Y) \\
&=&B_{n}^{-1}\sum_{b=1}^{B_{n}}\left\{ I(\hat{m}^{(b)}=\hat{m})-E(I(\hat{m}%
^{(b)}=\hat{m})\Vert Y)\right\} +P(\hat{m}^{(b)}=\hat{m}\Vert Y).
\end{eqnarray*}%
The final term on the right-hand side does not depend on $b$ and coincides
with $\tilde{r}(\hat{m},\hat{m})$. It converges to $1$ in probability as
shown in the proof of Part (a). Now, conditional on $Y$, the first term on
the right-hand side is the average of independent random variables that are
centered versions of Bernoulli-distributed variables. Hence, conditional on $%
Y$, this term has zero expectation and variance bounded by $\left(
4B_{n}\right) ^{-1}$, which goes to zero as $n\rightarrow \infty $. An
application of Chebyshev's inequality shows that this term hence converges to
zero in probability. This completes the proof of the first claim in case $%
B_{n}$ diverges to infinity. The case of a general sequence $B_{n}$ now
follows by a standard subsequence argument. The proof is thus complete as
the last claim reduces to the last claim in Part (a).
\ \rule{0.5em}{0.5em}
\end{proof}

Proposition~\ref{p1} obviously strengthens and extends Theorem 1 of %
\citet{Li18a}. However, this result also raises some obvious concerns: In
the limit, the coverage probability of the model confidence set (MCS)
defined by $\tilde{m}_{L}$ and $\tilde{m}_{U}$ ($\hat{m}_{L}$ and $\hat{m}%
_{U}$, respectively) is guaranteed to equal one (and not only to be $\geq
1-\alpha $), and the `length' of this MCS equals zero. In particular,
asymptotically this MCS coincides with the set $\left\{ \hat{m}\right\} $
and thus provides no more information than the point-estimate $\hat{m}$,
which the MCS was intended to improve. Interestingly, the fact that the MCS
asymptotically reduces to $\left\{ \hat{m}\right\} $ under the assumptions
of Proposition 1 seems to have informally been noticed by \citet{Li18a}, at
least in the context of model selection by the adaptive Lasso; see the
discussion in Section 1.5 of the supplementary material of that paper.
Apparently, however, this did not raise any red flags. [We also note that --
in light of the proof of Part (a) given above -- the proof of Theorem 1 in %
\citet{Li18a} seems to be much too complicated; in fact, it seems to be
incorrect: For example, the final inequality in the first display in the
proof of Theorem 1 appears to confuse the complement of $\left\{ \hat{m}%
_{L}\subseteq \hat{m}\right\} $ with $\left\{ \hat{m}_{L}\supsetneqq \hat{m}%
\right\} $.]

It turns out that a \emph{fixed-parameter} asymptotic analysis, as used in
Proposition~\ref{p1} and also in Theorem 1 of \citet{Li18a}, is not
appropriate here. In fixed-parameter asymptotics, the true parameter $\theta 
$ is kept fixed while sample size $n$ increases to infinity. At this point,
it seems fitting to repeat a warning issued a while ago by 
\citet[p.
153]{Haj71a}:

\emph{``Especially misinformative can be those limit results that are not
uniform. Then the limit may exhibit some features that are not even
approximately true for any finite $n$ \dots'' }

In fact, in the presence of model selection, fixed-parameter asymptotic
results not only \emph{can}, but actually often \emph{will}, mislead as has
been amply documented in %
\citet{Lee03a,Lee02a,Lee02c,Lee02b,Lee06a,Lee04a,Poe09a,Poe09b,Poe09c,Poe11a}%
. This is also the case here: We shall provide a \emph{uniform} asymptotic
analysis which reveals the misleading character of the pointwise asymptotic
result above and in Theorem 1 of \citet{Li18a}. To this end we have to amend
the notation a bit: For given sample size $n$, true regression parameter $%
\theta \in {\mathbb{R}}^{p}$, and variance parameter $\sigma ^{2}\in
(0,\infty )$ we shall write $P_{n,\theta ,\sigma ^{2}}$ for $P$ to emphasize
its dependence on the quantities indicated; similarly we shall write $%
m^{\ast }(\theta )$ to denote the smallest correct model for the parameter $%
\theta $. [That is, $P_{n,\theta ,\sigma ^{2}}$ and $m^{\ast }(\theta )$
replace the symbols $P$ and $m^{\ast }$ used in \citet{Li18a}.] Conditions
(A.1) and (A.2) then become $P_{n,\theta ,\sigma ^{2}}(\hat{m}\neq m^{\ast
}(\theta ))=o(1)$ and $P_{n,\theta ,\sigma ^{2}}(\hat{m}^{(b)}\neq \hat{m}%
)=o(1)$, respectively. Compared to Proposition~\ref{p1} (and Theorem 1 of %
\citet{Li18a}) we also add a mild condition on the design matrix $X$.

\begin{proposition}
\label{p2} Assume (A.1) and (A.2) hold (for every $\theta \in {\mathbb{R}}%
^{p}$ and every $\sigma ^{2}\in (0,\infty )$), and assume that the sequence
of matrices $X^{\prime }X/n$ is bounded.

(a) Then for every $\sigma ^{2}\in (0,\infty )$ we have%
\begin{equation}
\inf_{\theta \in {\mathbb{R}}^{p}}P_{n,\theta ,\sigma ^{2}}(\tilde{m}%
_{L}\subseteq m^{\ast }(\theta )\subseteq \tilde{m}_{U})\quad =\quad o(1)
\label{inf_coverage}
\end{equation}%
as $n\rightarrow \infty $. In other words, the minimal coverage probability
of the MCS defined by $\tilde{m}_{L}$ and $\tilde{m}_{U}$ converges to zero.
Moreover, assuming (A.1) only, this statement continues to hold if $\tilde{m}%
_{L}$ and $\tilde{m}_{U}$ are both replaced by $\hat{m}$.

(b) Part (a) also holds if we replace $\tilde{m}_{L}$ by $\hat{m}_{L}$ and $%
\tilde{m}_{U}$ by $\hat{m}_{U}$ (where the number of bootstrap replications $%
B$ may depend on $n$).
\end{proposition}

\begin{proof}
(a) Without loss of generality we may assume that the probability space
underlying $P_{n,\theta ,\sigma ^{2}}$ is given by $\mathbb{R}^{n}\times 
\mathcal{M}_{all}$, where $\mathbb{R}^{n}$ acts as the sample space for the
data $Y$ and $\mathcal{M}_{all}$ is the set of all models (and the $\sigma $%
-field is the product of the Borel-$\sigma $-field on $\mathbb{R}^{n}$ with
the power-set of $\mathcal{M}_{all}$). The probability measure $P_{n,\theta
,\sigma ^{2}}$ can then be viewed as given by 
\begin{equation}
P_{n,\theta ,\sigma ^{2}}(A)=\sum\nolimits_{m\in \mathcal{M}%
_{all}}\int\nolimits_{\mathbb{R}^{n}}I_{A}(y,m)K_{n}(y,m)dQ_{n,\theta
,\sigma ^{2}}(y)  \label{rep}
\end{equation}%
where $Q_{n,\theta ,\sigma ^{2}}$ is the probability measure on $\mathbb{R}%
^{n}$ defined by the $N(X\theta ,\sigma ^{2}I_{n})$-distribution, $%
K_{n}(y,m) $ is the bootstrap distribution on $\mathcal{M}_{all}$ (i.e.,
corresponds to a Markov-kernel from $\mathbb{R}^{n}$ to $\mathcal{M}_{all}$%
), and $I_{A}$ is the indicator function of $A$. Now, choose an arbitrary
constant $\gamma >0$ and a vector $\theta ^{(0)}\in {\mathbb{R}}^{p}$ such
that at least one coordinate of $\theta ^{(0)}$ equals zero. By Proposition~%
\ref{p1}, we have $P_{n,\theta ^{(0)},\sigma ^{2}}(\tilde{m}_{L}=\tilde{m}%
_{U}=m^{\ast }(\theta ^{(0)}))=1+o(1)$. For each $n$, choose $\theta
^{(n)}\in {\mathbb{R}}^{p}$ so that $m^{\ast }(\theta ^{(n)})\neq m^{\ast
}(\theta ^{(0)})$ and so that $\Vert \theta ^{(n)}-\theta ^{(0)}\Vert \leq
\gamma /\sqrt{n}$. (E.g., obtain $\theta ^{(n)}$ from $\theta ^{(0)}$ by
adding $\gamma /\sqrt{n}$ to one of the coordinates of $\theta ^{(0)}$ that
are equal to zero.) Because the sequences of measures $P_{n,\theta
^{(0)},\sigma ^{2}}$ and $P_{n,\theta ^{(n)},\sigma ^{2}}$ are mutually
contiguous (see below), it follows that $P_{n,\theta ^{(n)},\sigma ^{2}}(%
\tilde{m}_{L}=\tilde{m}_{U}=m^{\ast }(\theta ^{(0)}))=1+o(1)$. By
construction, we have $m^{\ast }(\theta ^{(0)})\neq m^{\ast }(\theta ^{(n)})$%
, and hence $P_{n,\theta ^{(n)},\sigma ^{2}}(\tilde{m}_{L}\subseteq m^{\ast
}(\theta ^{(n)})\subseteq \tilde{m}_{U})=o(1)$. [Mutual contiguity of $%
P_{n,\theta ^{(0)},\sigma ^{2}}$ and $P_{n,\theta ^{(n)},\sigma ^{2}}$ is
seen as follows: $Q_{n,\theta ^{(0)},\sigma ^{2}}$ and $Q_{n,\theta
^{(n)},\sigma ^{2}}$ are well-known to be mutually contiguous, see, e.g.,
Lemma A.1 in \citet{Lee02c} and note that this lemma, while given for the
case $\sigma ^{2}=1$, obviously extends to any $\sigma ^{2}>0$. The claim
then follows from Lemma 3.6 in \citet{Lee02c}.]

(b) Completely analogous, where now $\mathcal{M}_{all}$ is replaced by the $%
B $-fold Cartesian product $(\mathcal{M}_{all})^{B}$, $I_{A}(y,m)$ and $%
K_{n}(y,m)$ are replaced by $I_{A}(y,m_{1},\ldots ,m_{B})$ and $%
K_{n}(y,m_{1},\ldots ,m_{B})$, respectively, and where the latter represents
the bootstrap distribution of a bootstrap sample of size $B$.
\ \rule{0.5em}{0.5em}
\end{proof}

Proposition~\ref{p2} shows that the proposed MCS does not perform as desired
in that the worst-case coverage probability tends to zero, and not to the
nominal coverage probability $1-\alpha $. It also shows that the
fixed-parameter asymptotic setting used in Theorem 1 of \citet{Li18a} and in
Proposition~\ref{p1} paints a misleadingly optimistic picture of the actual
performance of the MCS procedure. This is in line with similar findings in
the context of inference post-model-selection reported earlier in %
\citet{Lee03a,Lee02a,Lee02c,Lee02b,Lee06a,Lee04a,Poe09a,Poe09b,Poe09c,Poe11a}%
. In fact, the final claim in Part (a) of Proposition~\ref{p2} can be easily
read-off from results concerning model selection probabilities already given
in these references (for a start see, e.g., Section 2.1 of \citet{Lee03a}).

The proof of Proposition~\ref{p2} reveals that this result continues to hold
if the infimum over $\theta \in {\mathbb{R}}^{p}$ in (\ref{inf_coverage}) is
further restricted, even in a sample-size dependent (e.g., shrinking) way,
as long as the infimum at sample size $n$ is taken over a set containing a
point $\vartheta ^{(n)}$ so that the sequence $\vartheta ^{(n)}$ has the
same properties as the sequence $\theta ^{(n)}$ that is used in the proof.
Also, under assumptions (A.1) and (A.2), a result similar to Proposition~\ref%
{p2} continues to hold in sufficiently smooth parametric models (as long as
the contiguity argument used in the proof goes through). This, of course,
also covers regression models with stochastic regressors.

Proposition 2 is an asymptotic result. To illustrate its import for
finite-sample situations we performed a small simulation study extending
some of the results presented in \citet{Li18a}. Data was generated as
described in Section 5 of \citet{Li18a}. In particular, we consider scenario
(c) of Section 5.1 of that paper with the only difference that we consider
not only one but several values for the parameter vector $\theta $. More
precisely, we set $B=1000$, $n=300$, $p=15$, $p^{\ast }=6$, $\sigma ^{2}=1$, 
$\rho =0.5$, $\theta _{1}=\dots ,\theta _{p^{\ast }-1}=1$, and vary $\theta
_{p^{\ast }}$ through the values $\theta _{p^{\ast
}}=0.05,0.1,0.2,0.3,0.4,0.5,1,2$. Note that $m^{\ast }=\{1,\dots ,p^{\ast
}\} $. The nominal confidence level was set to 90\%, i.e., $\alpha =0.1$.
The SCAD, Lasso, minimum BIC, and minimum AIC were used as model selectors.
We used the residual bootstrap for SCAD, BIC, and AIC, and the modified
residual bootstrap for Lasso. Descriptions of the bootstrap algorithms are
provided in Section 1.3 of the supplementary material of \citet{Li18a}.

For each of the eight different values of $\theta $ listed above and for
each of the four model selectors, we generated 200 response vectors $Y^{(i)}$%
, obtained $(\hat{m}_{L}^{(i)},\hat{m}_{U}^{(i)})$, and computed an estimate
for the coverage probability via 
\begin{equation*}
CP_{\theta }=(1/200)\sum_{i=1}^{200}I(\hat{m}_{L}^{(i)}\subseteq m^{\ast
}\subseteq \hat{m}_{U}^{(i)}).
\end{equation*}%
Let $CP_{\ast }$ denote the minimum of the $CP_{\theta }$-values when $%
\theta $ varies through the eight values mentioned above. This results in an
estimate of a (loose) upper bound for the worst-case coverage probability
for each of the four model selection procedures considered. [Of course, we
could have searched more thoroughly over the parameter space for $\theta $
to get a smaller upper bound for the worst-case coverage probability, but
this would require a much higher investment in computational resources
(including a second-stage simulation correcting downward bias resulting from
searching for the minimum). Searching only over eight values of $\theta $,
as we do here, incurs only a, for our purpose, negligible bias. The rough
upper bound $CP_{\ast }$ we obtain is good enough to make our point.] We
summarize the values of $CP_{\ast }$ for the four procedures in the
subsequent table:

\begin{center}
\begin{tabular}[t]{|l|l|l|l|l|}
\Hline
& $BIC$ & $SCAD$ & $AIC$ & $LASSO$ \\ \hline
$CP_{\ast }$ & $0.45$ & $0.16$ & $0.92$ & $0.94$ \\ \hline
\end{tabular}
\end{center}

\vspace{0.5cm}
The result for $CP_{\ast }$, when the consistent model selection procedure
minimum BIC is used, is in good agreement with Proposition 2, as $CP_{\ast }$
is much smaller than the nominal confidence level $0.9$. [In additional
simulations for model selection by minimum BIC using other values of $\theta 
$ we even have found values of $CP_{\theta }$ as small as 0.06.] Also, for
SCAD the observed $CP_{\ast }$ is much smaller than $0.9$. [Following %
\citet{Li18a} we have chosen the tuning parameter via crossvalidation. %
\citet{Li18a} do not show that this results in a consistent model selection
procedure, and probably it does not. However, we have not investigated this
issue. If the so-tuned SCAD is not a consistent model selection procedure,
then Proposition 2 strictly speaking does not apply. Nevertheless, the
simulations show that the worst-case coverage probability can be way below
the nominal one.] For the minimum AIC procedure, as well as for Lasso, the
observed $CP_{\ast }$ is about right. Since minimum AIC is not a consistent,
but rather a conservative, model selection procedure, Proposition 2 does not
apply; probably the same is true for Lasso given the tuning used. [Again the
tuning parameter for Lasso has been chosen via crossvalidation. \citet{Li18a}
do not show that this results in a consistent model selection procedure, and
probably it does not. We have not investigated this issue any further.] Of
course, this is not to say that the worst-case coverage probability for both
of these procedures is in any way guaranteed to be close to the nominal one;
it could also be that a more exhaustive search over the parameter space
would have turned up a much lower worst-case coverage probability also for
these procedures.

The above discussion shows that the proposed MCS has substantial defects
when used with consistent model selection procedures. It leaves open the
possibility that the MCS procedure suggested by \citet{Li18a} may have merit
for some conservative model selection procedures; however, as no theoretical
support for this is presented in \citet{Li18a}, this still needs to be
investigated. Another possible route for research is to consider alternative
asymptotic scenarios where, e.g., the number of parameters increases with
sample size. Problems of prediction following model selection have been
successfully tackled in such a framework, see \citet{Lee05a,Lee09a,Ste18a}.

\bibliographystyle{biom}

\end{document}